\documentclass{article}
\usepackage{blindtext}
\usepackage[a4paper, total={6in, 8in}]{geometry}

\usepackage{graphicx}%
\usepackage{multirow}%
\usepackage{amsmath,amssymb,amsfonts}%
\usepackage{amsthm}%
\usepackage{mathrsfs}%
\usepackage[title]{appendix}%
\usepackage[dvipsnames]{xcolor}%
\usepackage{textcomp}%
\usepackage{manyfoot}%
\usepackage{booktabs}%
\usepackage{algorithm}%
\usepackage{algorithmicx}%
\usepackage{algpseudocode}%
\usepackage{listings}%
\usepackage{lmodern}
\usepackage{authblk}

\usepackage[colorlinks=true, allcolors=blue]{hyperref}

\newcommand{\abs}[1]{\left\lvert #1\right\rvert}

\newcommand{\inv}{^{-1}}

\newcommand{\eps}{\varepsilon}

\usepackage{algorithm}
\usepackage{algpseudocode}

\renewcommand{\subset}{\subseteq}

\newcommand{\Number}[1]{\mathbb{#1}}

\newcommand{\R}{\Number{R}}

\newcommand{\Z}{\Number{Z}}
\newcommand{\X}{\Number{X}}
\renewcommand{\S}{\Number{S}}

\renewcommand{\phi}{\varphi}
\renewcommand{\epsilon}{\varepsilon}

\newcommand{\MT}{\mathcal{MT}}
\newcommand{\LMT}{\mathcal{LMT}}

\renewcommand{\subset}{\subseteq}

\newtheorem{theorem}{Theorem}[section]%
\newtheorem{cor}[theorem]{Corollary}%

\newtheorem{lemma}[theorem]{Lemma}

\newtheorem{definition}[theorem]{Definition}%

\date{}

\title{A Distance for Geometric Graphs via the Labeled Merge Tree Interleaving Distance}

\author[3]{Erin Wolf Chambers}
\author[1, 2]{Elizabeth Munch}
\author[1]{Sarah Percival}
\author[1]{Xinyi Wang}

\affil[1]{Department of Computational Mathematics, Science, and Engineering, Michigan State University}
\affil[2]{Department of Mathematics, Michigan State University}
\affil[3]{Department of Computer Science and Engineering, University of Notre Dame}

\begin{document}
\maketitle

\begin{abstract}
    Geometric graphs appear in many real-world data sets, such as road networks, sensor networks, and molecules. 
We investigate the notion of distance between embedded graphs and present a metric to measure the distance between two geometric graphs via merge trees. 
In order to preserve as much useful information as possible from the original data, we introduce a way of rotating the sublevel set to obtain the merge trees via the idea of the directional transform. 
We represent the merge trees using a surjective multi-labeling scheme and then compute the distance between two representative matrices. 
We show some theoretically desirable qualities and present two methods of computation: approximation via sampling and exact distance using a kinetic data structure, both in polynomial time. 
We illustrate its utility by implementing it on two data sets.
\end{abstract}

\section{Introduction}

Geometric graphs are ubiquitous in representing real-world datasets, from road networks and consumer preferences to skeletons to our data set of study: outlines of leaves. 
A fundamental task in shape analysis is to determine whether or not two objects are the ``same,'' and if not, to provide a quantitative measure of how different they are so that we can cluster, compare, and simplify. 
Our work is motivated by the question of how we can define and efficiently compute a distance between any two graphs in a way that encodes information about the embedding itself; see \cite{Buchin2023} for a survey.
Note that this task fundamentally differs from more general graph distance measures where graphs would only be at distance 0 if they are isomorphic; in this setting, we need graphs to be both combinatorially isomorphic and have the same embedding under that isomorphism. 

To construct a distance, we turn to the idea of the directional transform \cite{PHT,SCHAPIRA199183,Curry2018HowMD, Ghrist2018,Munch2023}, which encodes a shape $\X \subset \R^d$ as a family of $\R$-valued functions $f_\omega: \X \to \R$ for each unit vector $\omega \in \S^{d-1}$.
Because many tools from Topological Data Analysis (TDA) are built to take as input such an $\R$-valued function, we can return a chosen topological signature for each unit vector. 
For example, the persistent homology transform (PHT) returns a persistence diagram for every $\omega$; and the Euler characteristic transform (ECT) returns an Euler curve for every $\omega$. 
The PHT and ECT are both injective on the space of finite simplicial complexes in $\R^2$ or $\R^3$ \cite{PHT}; i.e.~if $\text{PHT}(A) = \text{PHT}(B)$ (or $\text{ECT}(A) = \text{ECT}(B)$), then $A=B$.
This works more generally in the setting of $o$-minimal structures \cite{Ghrist2018,SCHAPIRA199183} although we still need the information for every possible direction to promise injectivity. 
However, in practice, we can only compute the ECT or PHT for a finite number of angles. 
While injectivity still holds for a finite number of angles \cite{Curry2018HowMD, BeltonReconstruction}, these theoretical results require that the number of angles is exponential in the dimension of $\X$, and the directions chosen are specific to the particular input.  

Despite these limitations in theory, various applications have shown the effectiveness of the PHT and ECT, including but not limited to predicting the advancement of diseases based on tumor shapes \cite{CrawfordMonod, ECTTumorShboul}, differentiating cultivars based on leaf shapes \cite{Zhang2021}, measuring the morphological diversity of barley seeds \cite{BarleySeed}, and detecting variations in protein structures \cite{Tang2022}.
In practice, it is common to sample a large number of directions, particularly in the case of machine learning applications, and then test sensitivity to the number of directions chosen. 
\begin{figure}
    \centering
    \includegraphics[width=\linewidth]{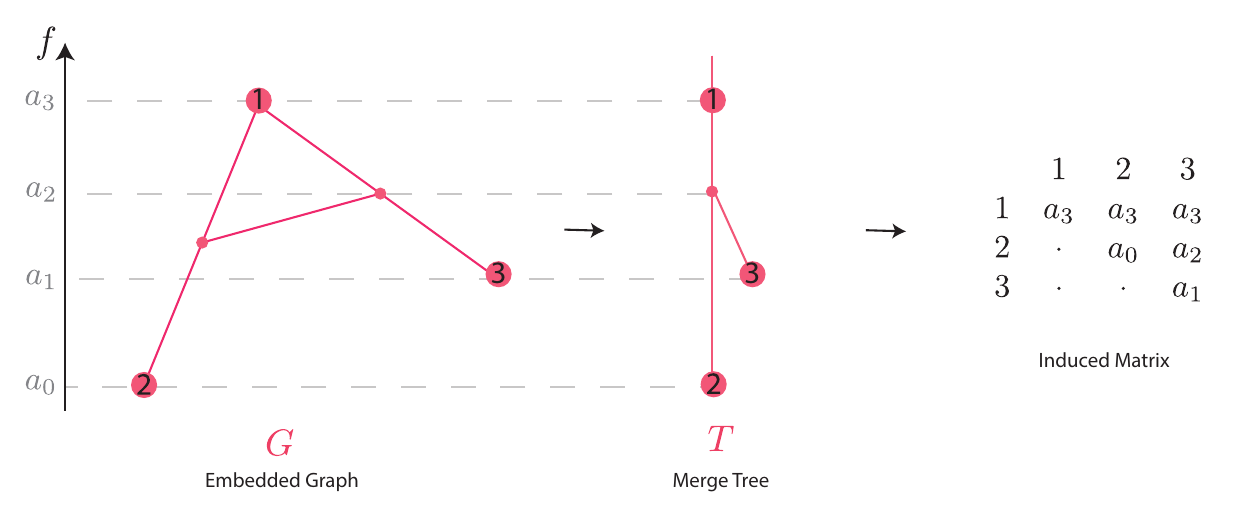}
    \caption{Left: A graph embedded in $\R^2$, along with a height function $f:V(G) \rightarrow \R$.  Middle: The merge tree of $(G,f)$, with three labels shown, each coming from a vertex of $G$. Right: The matrix $\mathcal{M}(T, f,\pi)$, as defined in Section~\ref{sec:dfn-labeledtree}.}
    \label{fig:GraphtoTree}
\end{figure}
In this paper, we study directional transforms specific to embedded graph data,  where rather than returning persistence diagrams or Euler characteristic curves, we consider a merge tree instead \cite{MergeTreeGH, MergetreeInterleaving, DmitriyMergeTree, PascucciMergeTree}. 
Given a space and an $\R$-valued function, merge trees are used to track the connectivity of sublevel sets in the space; see Fig.~\ref{fig:GraphtoTree} for an example.
Merge trees are commonly employed in TDA and visualization in a wide range of applications, such as  
analyzing combustion processes \cite{Bennett2, Mascarenhas2011}, function simplification \cite{Bauer2012}, and many more~\cite{Yan2021}.
Compared to many other topological signatures, merge trees are faster to compute, and track connectivity of sublevel sets, which persistent homology does not. 
However, they only track $0-$th dimensional structure and so higher dimensional information is lost.

\begin{figure}\centering
\includegraphics[width=0.8\linewidth]{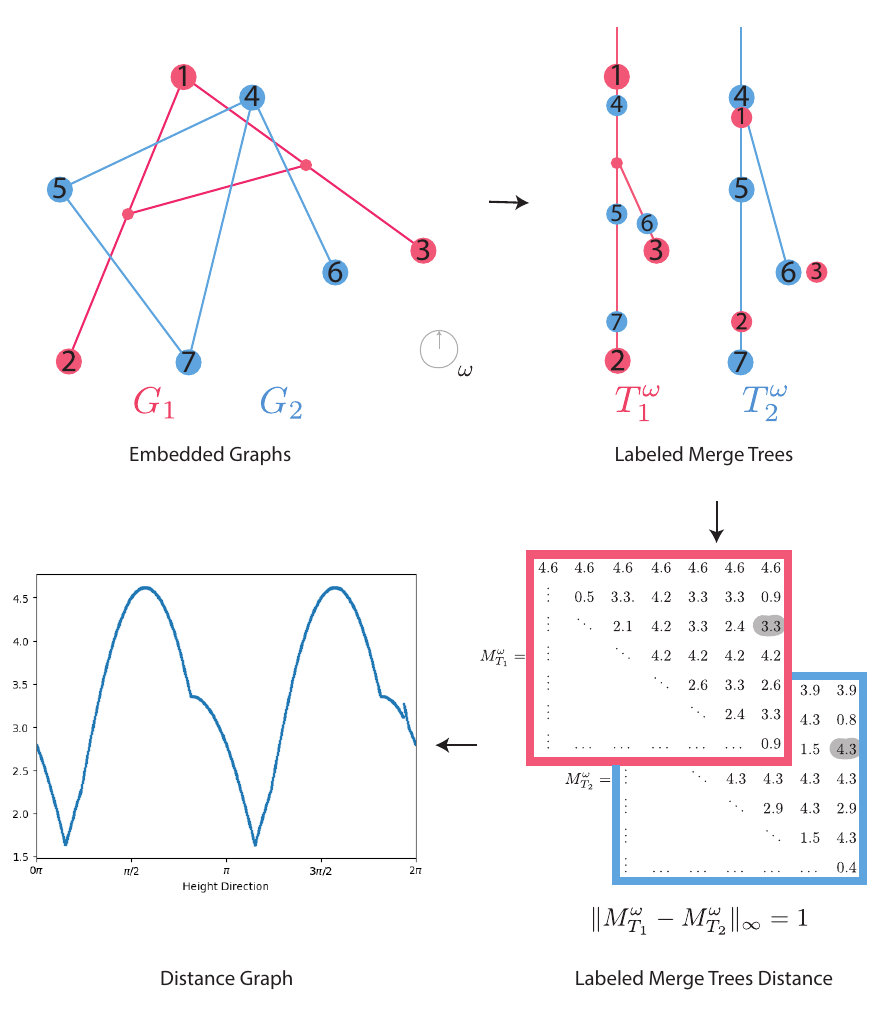}
\caption{Our method of comparing two embedded graphs, which takes every height direction $\omega \in S^1$, converts each graph to a labeled merge tree, then summarizes the trees in a matrix representation in order to calculate a distance.}
\label{fig:fullProcedure}
\end{figure}
Since our goal is a computable distance between merge trees and because of the intractability of many available merge tree distances (see \cite{Yan2020} for a survey), we focus on the case of labeled merge trees \cite{Munch2019,MergetreeInterleaving} where computation is possible. 
A labeling of a graph is a map from $\{1,\cdots,n\}$ to the vertex set. 
For labeled merge trees, we require that the labeling is surjective on the leaves.
These can be labels available from the data, such as on a phylogenetic tree; alternatively, these labels can be added artificially. 

In Section \ref{sec:background}, we introduce the necessary background of merge trees and the directional transform.
We discuss our method and study its metric properties in Section \ref{sec:LMTT}.
We construct an algorithm to determine the merge tree of an embedded graph for any direction $\omega$, which we then convert to a  matrix representation; see Figure~\ref{fig:GraphtoTree}. We show that, as $\omega$ varies from $0$ to $2\pi$, there are only a linear number of distinct merge trees relative to the number of edges.
Based on the embedding of two graphs to be compared, we create a labeling scheme which matches the vertices that give birth to leaves, from a given direction, with the closest point in the other graph. 
These labels propagate to the merge trees, allowing for fast computation of the labeled merge tree interleaving distance from each direction using the matrix representations; we then compare the two graphs by integrating this distance over all directions. 
See Fig.~\ref{fig:fullProcedure} for a visual outline of the process, which we call the Labeled Merge Tree Transform (LMTT) distance.
In Section \ref{sec:algo}, we provide both theoretical and practical methods for computing and estimating the distance, respectively.
We incorporate a kinetic data structure to compute the exact distance.
The practicality of our distance is demonstrated by applying it to a \emph{Passiflora} leaves data set \cite{Chitwood2017} and a letter data set in Section \ref{sec: Application}. Our implementation is available at \href{https://github.com/elenawang93/Labeled-Merge-Tree-Transform}{github.com/elenawang93/Labeled-Merge-Tree-Transform}.

\section{Background}
\label{sec:background}
In this section, we briefly review the necessary background for merge trees and directional transforms.  
We refer the reader to recent books and surveys for a full overview~\cite{Oudot2015,DeyWang2021,Yan2021,Munch2023}.

\subsection{Merge Trees}
\label{ssec:mergetree}

Intuitively, a \textit{merge tree} is a $1$-dimensional topological invariant that tracks when components appear and merge in sublevel sets of a scalar function. 
It is closely related to other topological descriptors such as persistent homology \cite{EdelsbrunnerPH}, as the branches of merge trees track the same $0$-dimensional features that appear in the persistence diagram, and Reeb graphs \cite{Reeb}, which keep track of the evolution of components in the level set.

Our input data will be a graph with an $\R$-valued function. 
First, we treat a graph $G = (V,E)$ as a 1-dimensional stratified space, so that the edges are homeomorphic to copies of the interval $[0,1]$. 
We will additionally be given a continuous map $f:G \to \R$ where we will assume that the function is monotone on the vertices. 
This additional assumption means we can combinatorially represent the function by specifying the function on the vertex set $f:V(G) \to \R$ and then extend linearly to the edges.
We also note that whenever needed, we can subdivide an edge $e = (u,w)$ with a vertex $v$ with $f(u) \leq f(v) \leq f(w)$ to get an equivalent representation of the structure. 

Our next step will be to understand the merge tree of a given graph with a function $f:G \to \R$. 
A merge tree is defined as follows. 
\begin{definition}
\label{defn:merge_tree}
A \emph{merge tree} is a pair $(T,f)$ with $f:T \to \R \cup \{ \infty\}$ where $T$ is a rooted, finite, acyclic, connected graph (i.e.~a finite tree in the graph-theoretic sense with a distinguished root vertex, $r$), together with a real-valued function $f:V \to \R\cup \{\infty\}$. 
This function has the property that every non-root vertex has exactly one neighbor with a higher function value,
and $f(v) = \infty$ iff $v$ is the root $r$.
\end{definition}

We write $V(T)$ for all non-root vertices of the tree $T$. 
We call the non-root vertices with degree $1$ \emph{leaves} and denote the set of leaves $L(T)$.
Every vertex has an up- and a down-degree, given by the number of neighbors with higher (respectively lower) function value. 
A vertex $w$ is regular if its up- and down-degree are both equal to 1. 
In this case, denoting the neighbors of $w$ as $u$ and $v$, we can remove the vertex $w$ and add in the edge $u,v$ to get an equivalent representation of the merge tree in a topological sense. 
So we say two merge trees are \emph{combinatorially equivalent} if $T_1$ and $T_2$ are isomorphic as combinatorial trees up to this subdivision equivalence, ignoring the function value information.
We denote the space of all merge trees by $\MT$. 

The merge tree of a given graph with a function $f:G \to \R$ is defined as follows, following \cite{DmitriyMergeTree}. 
For graph $(G,f)$, define the epigraph of the function as $\text{epi}(f) = \{(x,y) \in G \times \R \mid y \geq f(x)\}$.
Define a function $\bar{f}:\text{epi}(f) \to \R$ by $(x,y) \mapsto y$.
Say that two points $(x,y)\sim_{\bar{f}} (x',y')$ are equivalent iff $y = y'$ (i.e.~$\bar{f}(x,y) = \bar{f}(x',y')$) and they are in the same connected component of the levelset $\bar{f}^{-1} (a)$.
Then the merge tree of $(G,f)$, denoted $T(G,f)$,  is the quotient space $\text{epi}(f) / \sim_{\bar{f}}$.\footnote{The experts will notice that this is the Reeb graph of  $(\text{epi}(f), \bar{f})$.}
Note that the points of $T(G,f)$ inherit the function $f$ since we can set $f([x]) = f(x)$ which is well defined by construction of $\sim_f$.
See Figure~\ref{fig:GraphtoTree} for an example. 

Note that this definition is subtly different from the one used in the visualization literature (see e.g.~\cite{Yan2021}).  
Define an equivalence class on the points of $G$ by setting $x \sim y$ iff $f(x) = f(y)$,  and $x$ and $y$ are in the same connected component of $f^{-1}(-\infty,f(x))$.
Then the alternative merge tree is the quotient space $X/\sim$. 
While the result is nearly identical, the difference is the root vertex will have a function value equal to the global maximum of $f$; in particular in the case of compact $X$ this results in a finite function value for the root which goes against Definition \ref{defn:merge_tree}.
However, this causes issues in the case of some distance computations, specifically the interleaving distance \cite{DmitriyMergeTree,Munch2019,MergetreeInterleaving}, where an infinite tail (i.e.~$f(r) = \infty$) is needed for technical reasons. 
In this paper, we will use the infinite tail definition due to our focus on the use of distances.
However, we note that we could use the first definition with some modification as in \cite{DmitriyMergeTree}; or alternatively, we could attach a root vertex $r$ to the global maximum vertex in the alternative merge tree and set $f(r) = \infty$ to resolve the issue.

We will next introduce labeled merge trees. 
\begin{definition}
Fix an integer $n \in \Z_{\geq 0}$  let $[n] = [1,\ldots,n]$. 
An \emph{$n$-labeled merge tree} is a merge tree $(T,f)$ together with a map to the non-root vertices $\pi: [n]\rightarrow V(T)$ that is surjective on the set of leaves. 
We denote a labeled merge tree with the triple $(T, f, \pi)$, and the space of all $n$-labeled merge trees by $\LMT_n$.
\end{definition}
Note that this requirement has a label on all leaves, but might involve having labels on interior vertices, including those resulting from the subdivision of an edge.
See Fig.~\ref{fig:GraphtoTree} for an example.

Finally, assume we are given a labeling (no longer requiring surjectivity in any form) on an input graph $\pi:[n] \to V(G)$, where we recall that we can subdivide an edge as needed to have this labeling be defined on the vertex set. 
Then we can use this to induce a labeling $\bar \pi:[n] \to V(T)$ on the merge tree $T = T(G,f)$ by setting $\bar \pi(i)$ to be the equivalence class of $\pi(i)$ under the quotient space construction. 
Note that this induced labeling does not have the required surjectivity on the leaves without additional assumptions on $\pi$; later in Cor.~\ref{cor:surjectiveLabel} we will provide assumptions to make this happen.

\subsection{Labeled Merge Tree Distance}
\label{ssec:labeledMergeTreeDistance}
\label{sec:dfn-labeledtree}

Finding metrics to compare different merge trees has recently been a focus of study, given their wide use as a simple topological summary of scalar field data.
In this paper, we focus on the interleaving distance, which has been broadly used to compare persistence diagrams, with a range of theoretical and practical guarantees in various settings.
Computing the interleaving distance on merge trees is NP-hard \cite{MergeTreeGH, InterleavingComplexity}, although there has been work done on heuristic methods for certain types of data~\cite{Yan2020}.
The reason we are interested in labeled merge tree data here is that we can utilize the $L_\infty$-cophenetic metric~\cite{Munch2019,MergetreeInterleaving}, also known as the labeled interleaving distance, since computation becomes tractable in this setting.

In order to state the interleaving distance,  we define a natural conversion of a merge tree into a matrix as follows. 
Let $\text{LCA}(v, w)\in T$ denote the \textit{lowest common ancestor} of vertices $v$ and $w\in V(T)$, that is the vertex with the lowest function value $f(v)$ that has both $v$ and $w$ as descendants, where we define each node to be a descendant of itself. 
The \emph{induced matrix of the labeled merge tree} $(T, f, \pi)$
is $\mathcal{M}(T, f,\pi)\in\R^{n\times n}$, where
$\mathcal{M}(T, f,\pi)_{ij}=f(\text{LCA}(\pi(i), \pi(j)))$.
See Fig.~\ref{fig:GraphtoTree} for an example.
We can now compute the distance between two labeled merge trees with the same set of $n$ labels by taking the $L_\infty$ norm of the difference between the two induced matrices; see Fig.~\ref{fig:enter-label}. 
More formally, given a matrix $A$, recall that the $L_\infty$ norm of $A$ is $\|A\|_\infty = \max\abs{a_{ij}}$. 
Then the distance is defined as follows.
\begin{definition}
\label{defn:labeledInterleavingDist}
For a fixed $n>0$ and two labeled merge trees in $\LMT_n$, the \emph{labeled interleaving distance} is
\[
d((T_1, f_1,\pi_1),(T_2, f_2,\pi_2)) :=  \| \mathcal{M}(T_1, f_1,\pi_1)-\mathcal{M}(T_2, f_2,\pi_2)\|_\infty.
\]
\end{definition}

\begin{figure}
    \centering
    \includegraphics[width = \linewidth]{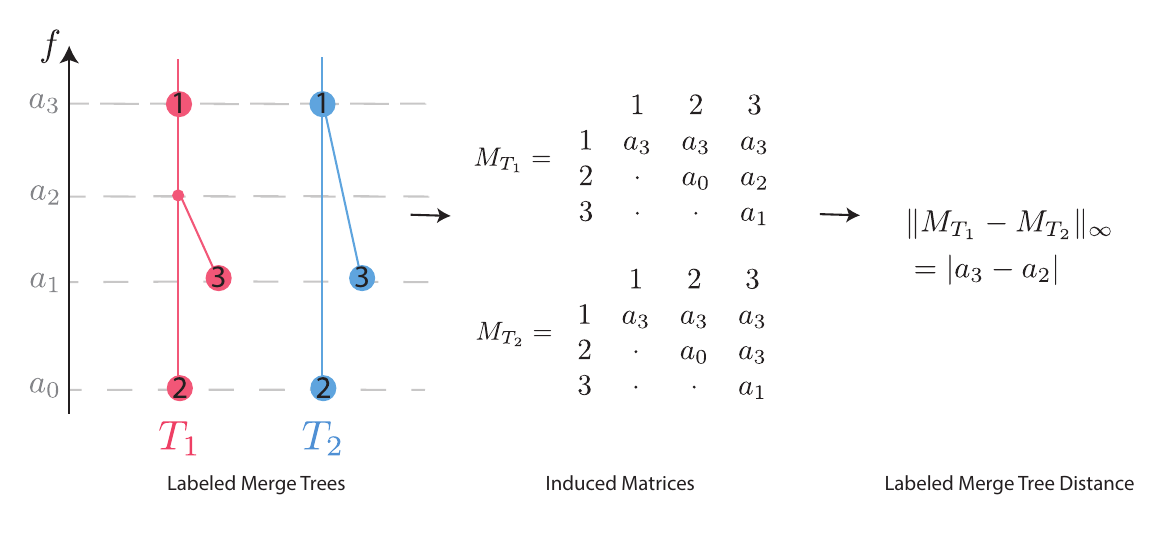}
    \caption{Two merge trees with a common set of labels, along with their induced matrices and the labeled merge tree distance between them.}
    \label{fig:enter-label}
\end{figure}

It is worth noting that this yields a metric on the space of merge trees with a common label set \cite{Munch2019,MergetreeInterleaving}, and further is an interleaving distance \cite{deSilva2018} for a particular choice of category. 

\subsection{Directional Transform on Embedded Graphs}  
\label{sec:DT}

Directional transforms are a relatively recent development arising from computing topological features, such as persistent homology or the Euler characteristic, of an object embedded in Euclidean space using varying directions to construct a family of input height functions.
This transform first appeared in the TDA literature in \cite{PHT}, where Turner \textit{et al.} showed that the  Euler characteristic transform (ECT) and persistent homology transform (PHT) computed with all possible directions $\S^{d-1}$ are both injective on the space of simplicial complexes in $\R^d$.
Later work \cite{Curry2018HowMD} proved that we need a large but finite set of directions for the PHT and ECT to preserve the topology fully.  
More recent work gave a quadratic time algorithm to explicitly compute all necessary directions for persistent homology to reconstruct a graph embedded in $\R^2$~\cite{BELTON2020101658}.
Additionally, using the machinery of Euler calculus \cite{Schapira1995}, the injectivity result was extended concurrently by \cite{Ghrist2018} and \cite{Curry2018HowMD} to nice enough subsets in any ambient dimension.
We take inspiration from many of these ideas but use the labeled interleaving distance on merge trees, given previous work that successfully used a merge tree transform to encode data, but used a different distance for the merge trees~\cite{Batakci2023}.

Our input data will be embedded graphs. 
As in Sec.~\ref{ssec:labeledMergeTreeDistance}, we treat a graph $G = (V,E)$ as a 1-dimensional stratified space, so that the edges are homeomorphic to copies of the interval $[0,1]$. 
We say a graph $G = (V,E)$ is \textit{embedded} in $\R^d$ if it has a continuous function $\phi : G\rightarrow \R^d$ for $d\geq 2$, that is homeomorphic onto its image. 
Note that this means each vertex in $V$ is mapped to a point, and each edge in $E$ is mapped to a curve in $\R^d$ so that the graph structure is maintained. 
For ease of notation, we often denote this information as a pair $(G,\phi)$. 
In this paper, we restrict ourselves to a graph embedded in $\R^2$ where each edge is a straight line between its endpoints, which we refer to as a \textit{geometric graph}.\footnote{Note that the term geometric graph can have different definitions in the literature; our definition matches the usage in~\cite{Goodman2018} with the additional restriction that the edges cannot cross.}

Fix a unit vector direction $\omega \in \S^1$. 
We can define the height function for direction $\omega$, $f_\omega: G \to \R$, on the embedded graph $(G,\phi)$ by setting $f_\omega(x) = \langle \phi(x),\omega \rangle$.
For a fixed $a\in\R$, the sublevel set at $a$ is just the subgraph $G_a$ induced by the vertices $\{v \mid f_\omega(v) \leq a\}$.
Then the merge tree for direction $\omega$, denoted $T^\omega(G)$, is the merge tree of $(G,f_\omega)$. 
The merge tree transform is the function 
\begin{equation*}
\begin{matrix}
    MTT(G): & \S^1 & \longrightarrow & \MT\\
    & \omega & \longmapsto & T^\omega(G)
\end{matrix}
\end{equation*}

\begin{definition}
A vertex $v \in V(G)$ is an \textit{extremal vertex} if its embedding is \emph{not} in the convex hull of its neighbors $nbr(v) = \{w \mid vw \in E(G)\}$. 
The collection of extremal vertices $V' \subseteq V$ is called the \textit{extremal vertex set}.
\end{definition}
See Figure \ref{fig:Extremal}.
These vertices completely characterize the collection of vertices that result in a leaf in the merge tree for some direction, as shown in the following lemma. 
\begin{lemma}
\label{lem:extremalVertex}
 Fix a straight-line embedded graph $(G,\phi)$. 
 A vertex $v\in V$ gives birth to a leaf in a merge tree $T^\omega(G)$ for some $\omega$ if and only if $v\in V'$ (i.e.~$v$ is an extremal vertex).
\end{lemma}
\begin{proof}
For a straight-line embedded graph, the leaves of the merge tree can be characterized by the following property.
A vertex $v \in V(G)$ gives birth to a leaf in a merge tree for direction $\omega$ iff there is no neighbor $u$ of $v$ with function value $f_\omega(u) \leq f_\omega(v)$. 

\begin{figure}
    \centering
    \includegraphics[width = 0.9\textwidth]{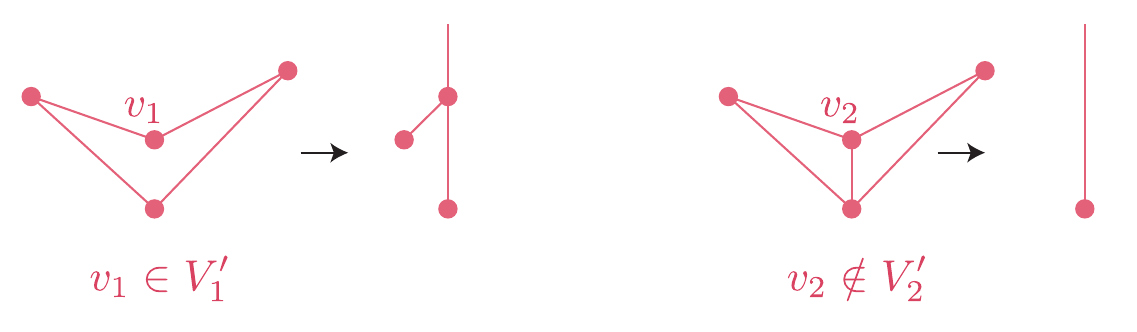}
    \caption{Examples of extremal and non-extremal vertices}
    \label{fig:Extremal}
\end{figure}

For the forward direction, assume $v$ gives birth to a leaf in a merge tree for some direction $\omega\in \S^1$. 
This means that for all $u \in nbr(v)$, $f_\omega(u) > f_\omega(v)$. 
Let $u'$ be the neighbor with the lowest function value among them. 
Then there is a line normal to $\omega$, which is the levelset at 
$\{ x \in \R^2 \mid f_\omega(x) = \tfrac{1}{2} (f_\omega(u')-f_\omega(v))\}$.
This line separates the point $v$ from $nbr(v)$, so it must also separate the convex hulls. 
Thus $v$ is not in the convex hull of its neighbors, so $v \in V'$.

For the other direction, assume we have an extremal vertex $v \in V'$. 
By the hyperplane separation theorem and because $v$ is its own convex hull, there exists a line separating $v$ and the convex hull of $nbr(v)$. 
Choose $\omega$ to be the normal to this line such that $f_\omega(v)$ is less than all its neighbors (else take the opposite normal). 
Then for the function $f_\omega$, all neighbors have function values greater than $v$, and thus $v$ gives birth to a leaf in direction $\omega$ as required.
\end{proof}

Note that this result immediately provides the assumptions discussed earlier to induce a labeling on the merge tree from a labeling on the graph. 

\begin{cor}
\label{cor:surjectiveLabel}
    Given a map $\pi:[n] \to V(G)$ which is surjective on the extremal vertex set, the induced labeling $\bar \pi:[n] \to V(T^\omega(G))$ is surjective on the leaves for any $\omega \in \S^1$.
\end{cor}

\section{Distance between Embedded Graphs}
\label{sec:LMTT}

In this section, we give the definition of a distance between embedded graphs via the labeled interleaving distance. 
At a high level, given two embedded graphs $G_1$, $G_2$, we will compute their merge trees $(T_1, f_1)$ and $(T_2, f_2)$ for height functions $f_\omega$ over all $\omega\in\S^1$. 
We then introduce a method to label the merge trees with a common label set using the geometric embedding in order to compute the labeled interleaving distance (Defn.~\ref{defn:labeledInterleavingDist}).  
We then define a distance between graphs by integrating over all $\omega\in \S^1$.
We conclude the section with a  discussion of properties of the distance.

\subsection{Surjective Labeling Scheme} \label{ssec:labelingscheme}
We start by constructing a labeling on the embedded graphs with the same label set $[n+m]$, which then induces a labelling on the merge trees for any fixed $\omega$.

Given embedded graphs $\phi_1:G_1\to \R^2$ and $\phi_2:G_2\to \R^2$, let $V_1$ and $V_2$ denote the vertex sets respectively. 
Recall that a vertex is extremal if it is not in the convex hull of its neighbors in the graph, and define the subsets $V_1' \subseteq V_1$ and $V_2'\subseteq V_2$ to be the respective extremal vertex sets. 
Enumerate these vertices by setting $V_1' = \{v_1,\cdots,v_n \}$ and $V_2' = \{w_{n+1}, \cdots, w_{n+m} \}$. 
Next, for each $v_i \in V_1'$, let $p_i$ be a point (choosing arbitrarily if non-unique) which is closest in the embedding of $G_2$ via the Euclidean distance; so 
\begin{equation*}
    \|f_1(v_i) - f_2(p_i)\| = \min_{x \in G}\|f_1(v_i) - f_2(x)\|.
\end{equation*}
Note that $p_i$ can be a vertex of graph $G_2$ or an interior point on an edge, in which case we will subdivide the edge and add $p_i$ to the vertex set $V_2$.
Similarly, find a point $q_j \in G_1$ which is closest to $w_j \in V_2'$, adding a $q_j$ to $V_1$ if necessary. 
This gives a labeling from the set $[n+m]$ to each geometric graph, which we denote as $\pi_1$ and $\pi_2$, given by 
\begin{equation*}
\begin{matrix}
\pi_1(i) =
\begin{cases}
     v_i & i \in [1,n]\\
     q_i & i \in [n+1,n+m]
\end{cases}
&\qquad &
\pi_2(i) =
\begin{cases}
     p_i & i \in [1,n]\\
     w_i & i \in [n+1,n+m].
\end{cases}
\end{matrix}
\end{equation*}

This labeling can be pushed forward to the merge tree for any direction, so we abuse notation and write $\pi_1:[n+m] \to V(T_1)$ and $\pi_2:[n+m] \to V(T_2)$ for the resulting labelings.
Note that by Cor.~\ref{cor:surjectiveLabel}, because the original labeling was defined for all extremal vertices, the resulting labeling is surjective on the leaves as is required by the interleaving distance definition.
The induced matrices of $T_1^\omega$ and $T_2^\omega$ therefore have size $(n+m)\times(n+m)$. 
See Fig.~\ref{TreeLabels} for an example of the resulting labeled merge trees and their matrices. 

\begin{figure}
\centering
\includegraphics[width = \textwidth]{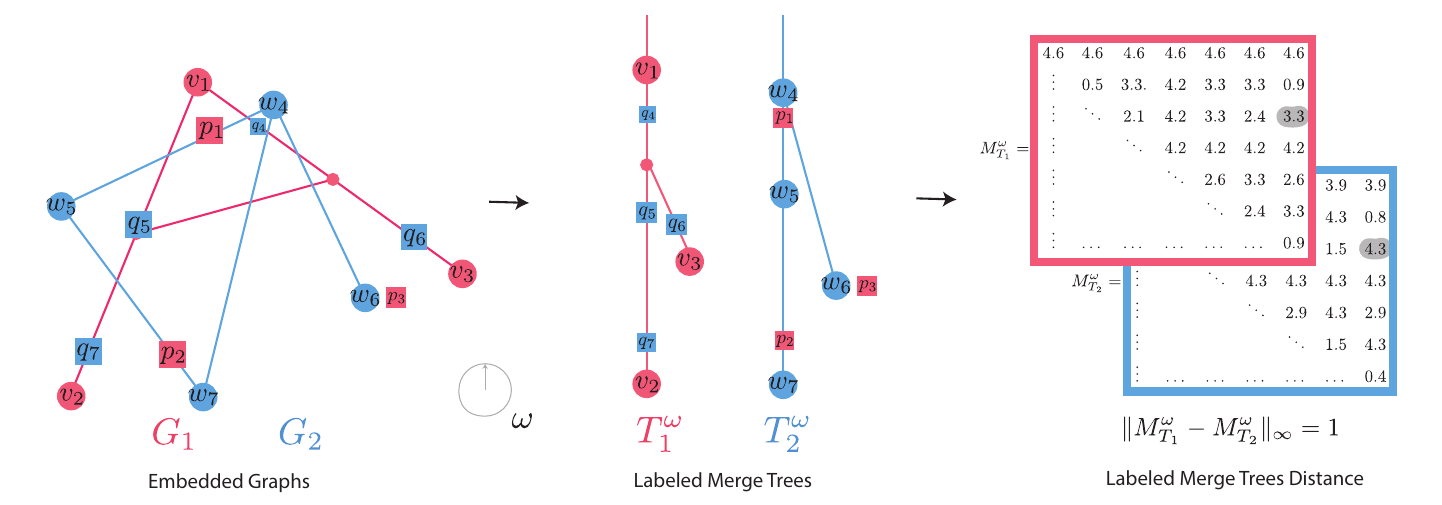}
\caption{Given $G_1$, $G_2$, and a direction $\omega$ as shown, we perform the surjective labeling scheme on $v\in V_1'\cup V_2'$.  Here, any vertex with a number inside is extremal, and so, generates a label in the merge trees.}
\label{TreeLabels}
\end{figure}

\subsection{The Labeled Merge Tree Transform Distance}

For a given pair of graphs, we fix the labels from the previous section, and use this to induce a labelling on the merge tree for each direction $\omega \in \S^1$ resulting in a pair of labeled merge trees $T_1^\omega$ and $T_2^\omega$.
Denote the induced matrices of $T_1^\omega$ and $T_2^\omega$ as $M_{T_1}^\omega$ and $M_{T_2}^\omega$ respectively. 
As seen in Defn.~\ref{defn:labeledInterleavingDist}, we have a distance between the labeled trees for every direction, 
$d(T_1^\omega,T_2^\omega)= \| M_{T_1}^\omega-M_{T_2}^\omega \|_\infty$.
To get a distance between the two collections of trees, we integrate over this distance as follows.

\begin{definition}
    \label{def:DMT}
    The \emph{labeled merge tree transform (LMTT) distance} between $G_1$ and $G_2$ is defined as
    \begin{equation*}
    D(G_1, G_2) = 
     \frac{1}{2\pi}\int_{0}^{2\pi} d(T_1^\omega, T_2^\omega) \, d\omega \ = 
     \frac{1}{2\pi}\int_{0}^{2\pi} \| M_{T_1}^\omega-M_{T_2}^\omega \|_\infty \, d\omega \
    \end{equation*}
    where the labeling of $G_1$ and $G_2$ is done as in Sec.~\ref{ssec:labelingscheme}.
\end{definition}

We use the term distance for this comparison measure, as we have not yet shown that this definition satisfies all the axioms of a metric.  Recall the following properties:

\begin{definition}
    Let $S$ be a set with a given function 
    $\delta:S\times S\rightarrow \R_{\geq 0}\cup\{\infty\}$, where $\R_{\geq 0}$ is the non-negative extended real line. 
    We define the following properties:
    \begin{enumerate}
        \item Finiteness: for all $x, y\in S$, $\delta(x, y)<\infty$.
        \item Identity: for all $x\in\mathbb{R}$, $\delta(x, x)=0$.
        \item Symmetry: for all $x, y\in S$, $\delta(x, y) = \delta(y, x)$.
        \item Separability: for all $x, y\in S$, $\delta(x, y) =0$ implies $x=y$.
        \item Subadditivity (Triangle Inequality): for all $x, y, z\in S$, $\delta(x, y)\leq \delta(x, z) + \delta(z, y)$.
    \end{enumerate}
    If $\delta$ satisfies finiteness, identity, and symmetry, it is called a \emph{dissimilarity function} \cite{Jardine1977}. 
If $\delta$ satisfies all five properties, then it is a \emph{metric}; it is a \emph{semi-metric} if it satisfies all but subadditivity. 
\end{definition}

\begin{theorem}
    The labeled merge tree transform $D$ is a dissimilarity function. 
    That is, it satisfies finiteness, identity, symmetry,

\end{theorem}

\begin{proof}
    For the proof, we assume we have inputs $(G_1,f_1)$ and $(G_2,f_2)$ with labels $\pi_1:[n+m] \to V(G_1)$ and $\pi_2:[n+m] \to V(G_2)$ throughout. 
    
    \textit{Finiteness:}
    Because we assume that the input graphs are finite, the image $f_2(G_1)\cup f_2(G_2)$ is contained in a finite bounding box. 
    This in turn means that there is a global bound $B$ such that $|f_\omega(x)-f_\omega(y)| \leq B$ for any $x,y \in f_2(G_1)\cup f_2(G_2)$ and $\omega \in \S^1$.
    Then the difference between any entries in $M_{T_1}^\omega$ and $M_{T_2}^\omega$ is also bounded by $B$, so $d(T_1^\omega,T_2^\omega) \leq B$ as well.
    This implies that the integral of the function is also bounded, so $D$ is finite.

    \textit{Identity}: 
    When computing $D(G, G)$, we make the tautological observation that the closest point of $G$ to $x \in G$ is, of course, $x$. 
    Thus, the extremal vertices will each have a pair of labels: $i$ and $i+n$ where $n = |V'(G)|$ and this labeling will be the same on both graphs. 
    Further, this means the resulting matrices and merge trees will be the same for each, so the distance for any direction is 0, and thus the integral is also 0. 

    \textit{Symmetry:} This is immediate from both the symmetry of the labeling and the symmetry of the $L_\infty$ norm of matrices. 
\end{proof}

\begin{figure}
\centering
\includegraphics[width = \textwidth]{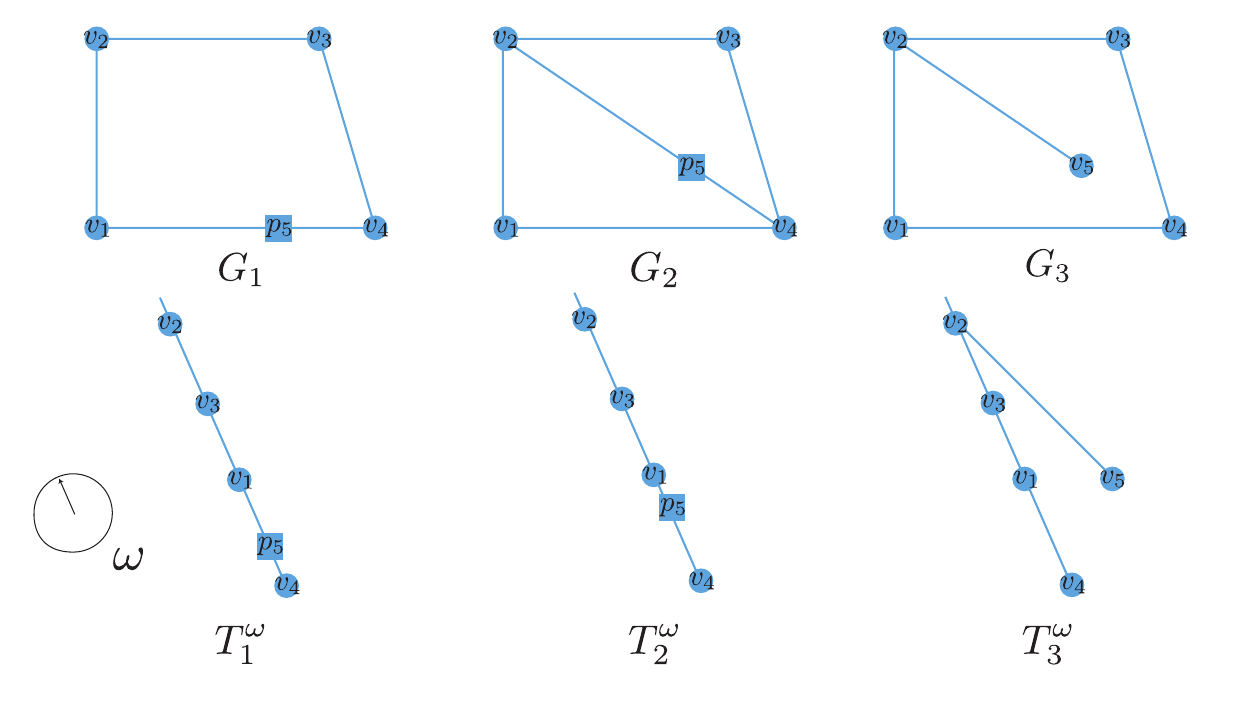}
\caption{
A counterexample for both separability and triangle inequality.  
}
\label{fig:BrokenGraphs}
\end{figure}

The LMTT is not a metric since it does not satisfy the separability or subadditivity properties. 
Both can be seen in the counterexample of Figure~\ref{fig:BrokenGraphs}.  
Here, $G_1$ is a simple trapezoid shape with vertices $v_1, \ldots, v_4$.  
These vertices are embedded identically for $G_2$ and $G_3$, but with the addition of one more edge from $v_2$ to $v_4$ in $G_2$, and an extra edge and vertex in $G_3$ that is a subset of the diagonal edge in $G_2$.
The merge tree of both $G_1$ and $G_2$ from any direction is a single edge, so the LMTT cannot hope to detect the difference between them, and separability fails.  
For triangle inequality, note that $d(G_1, G_2) = 0$ but that $d(G_1,G_3) > d(G_2,G_3)$ in any direction $\omega$, since $v_5$ can project with distance 0 to $p_5$ on $G_2$, but must project on $G_1$ to a point at positive distance.  Hence, $d(G_1,G_3) > d(G_2,G_3) + d(G_1,G_2)$. 

\subsection{Continuity}

In this section, we study continuity of the distance with respect to the varying parameter $\omega$. 
In the example of Figure \ref{fig:CriticalAngle}, we see that varying the direction leads to different merge tree structures for the same input graph.
When computing the distance between the embedded graphs, this means that we need to separate regions of $\S^1$ for which the underlying graph of the merge tree is the same. 
To find these regions, we first define a collection of \emph{critical angles}.

\begin{definition}
Given a graph $G=(V, E)$, for all $e\in E$, we define \emph{critical angles} to be the set of all normal vectors for each $e_i$. 
Specifically, if we denote $(v_i, v_j)=e$, then
\begin{equation}
    \label{eqn:crit}
\text{Crit}(G) := \left\{ \theta \in \S^1 \ \big| \ \vec{\theta} \cdot (\vec{v_i} - \vec{v_j}) = 0 \text{ for some } (v_i, v_j)=e \in E \right\}.
\end{equation}

\end{definition}

Since the number of edges of a finite graph is finite, we have a finite number of critical angles. 
Recall that two merge trees are combinatorially the same if the underlying trees are the same up to merging/subdivision of monotone edges and ignoring anything about the function values.
A crucial property of critical angles is that they define where the combinatorial merge tree could change given a continuous direction input from $\S^1$. 

\begin{figure}
\centering
\includegraphics[width = 0.6\linewidth]{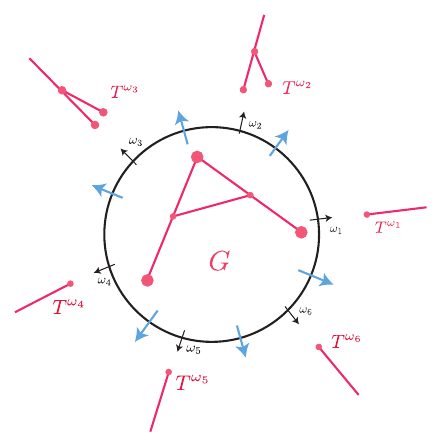}
\caption{For a given graph $G$, the critical angles are marked in blue. 
An example merge tree for each region of $\S^1$ between critical values is given. 
}
\label{fig:CriticalAngle}
\end{figure}

\begin{lemma}
\label{lem:CritAng}
Denote $\text{Crit}(G) = \{\theta_1,\theta_2,\cdots,\theta_n = \theta_0 \}$ ordered counterclockwise around the circle.
This forms a subdivision of $\S^1$ such that for any $\omega_1, \omega_2 \in (\theta_i,\theta_{i+1})$, $T^{\omega_1}$ and $T^{\omega_2}$ are combinatorially equivalent.
\end{lemma}
\begin{proof}
We will show that the set of vertices of $G$ giving birth to a connected component in direction $\omega_1$ and $\omega_2$ are the same. 
We will do the same for the set of vertices causing a merging. 
This implies that even if the function values on the merge tree are different, the combinatorial structure of the trees will be the same. 

Assume we have a birth causing vertex $v \in V(G)$ for direction $\omega_1$. 
By Lemma \ref{lem:extremalVertex}, this implies that $v$ is an extremal vertex, and so it is not in the convex hull of its neighbors. 
As in the proof of Lemma \ref{lem:extremalVertex}, $v$ causes a birth in direction $\omega_1$ iff all neighbors lie above the hyperplane perpendicular to $\omega_1$ through $v$. 
Because passing from  $\omega_1$ to $\omega_2$ does not cross any critical angles, the neighbors of $v$ must also lie above the hyperplane perpendicular to $\omega_2$ through $v$. 
Thus $v$ is still a birth causing vertex for direction $\omega_2$. 
This implies that the set of birth causing vertices are the same. 

Then, assume that $v$ merges two connected components in the direction $\omega_1$. 
This occurs iff there are two connected components in the sublevelset 
$f_{\omega_1} \inv (f_{\omega_1}(v)-\delta)$ for any small $\eps >0$. 
This sublevelset must be the same as $f_{\omega_2} \inv (f_{\omega_2}(v)-\eps)$, otherwise there would be a critical angle between $\omega_1$ and $\omega_2$. 
Thus $v$ is a merge vertex for $\omega_1$ iff it is also a merge vertex for $\omega_2$.
\end{proof}

For an example of this lemma see Figure \ref{fig:CriticalAngle} where the critical angles are shown in blue. 
Within the regions separated by the critical angles, the tree structure remains constant. 
However, passing across a critical angle, such as from $\omega_1$ to $\omega_2$, the tree structure can change.
Note that we could still have critical $\theta_i$ across which the combinatorial structure does not change, such as from $\omega_6$ to $\omega_1$, thus the lemma implies that the set of potential changes for the merge tree structure is a subset of $\text{Crit}(G)$.
With this in hand, we see that the distance function is piecewise continuous with discontinuities given by the set of critical angles. 
\begin{theorem}
Fix two geometric graphs $\phi_1:G_1 \to \R^2$, $\phi_2:G_2\to \R^2$, and define  $F(\omega)$ to be the $L_\infty$ norm between the two matrices obtained from the merge trees computed using direction; that is, $\omega\in\S^1$
\begin{equation*}
\begin{matrix}
    F:  & \S^1 & \rightarrow & \R\\
      & \omega & \mapsto & \|M_1^\omega-M_2^\omega\|_\infty.
\end{matrix}
\end{equation*}
Then, $F(\omega)$ is piecewise continuous with the set of discontinuities given by a subset of $\text{Crit}(G)$.
\end{theorem}

\begin{proof}
Let $\pi_1:[n] \to G_1$ and $\pi_2:[n] \to G_2$ be the labelings on the geometric graphs as defined in Sec.~\ref{ssec:labelingscheme}. 
Denote $v_i = \pi_1(i) \in G_1$ and $w_i = \pi_2(i) \in V(G_2)$ for all $i$. 
Now fix a pair $i,j$ and we first will show that between critical angles, every entry in $M_1^\omega-M_2^\omega$ is continuous. 

First, we will show the entry of $M_1^\omega[i,j]$ is continuous with respect to $\omega$. 
Between critical angles, by Lem.~\ref{lem:CritAng}, we know that the combinatorial tree remains the same.
This means that the vertex in the merge tree which is the LCA of $v_i$ and $v_j$ remains the same, and further that the vertex $v \in V(G)$ which causes this merge point is the same. 
Thus the entry $M_1^\omega[i,j] = f_\omega(v)$, which is continuous.
The same argument holds for $M_1^\omega[i,j]$, thus the difference between the matrices is continuous as required. 
Finally, since the $L_\infty$ norm of a matrix is continous for continuously varying entries, we have that $F$ is continuous between critical angles.
\end{proof}

\section{Implementation}
\label{sec:algo}

Here, we discuss our implementation of the labeled merge tree transform distance computation, starting with two embedded graphs as input and returning either an exact value or approximation of the final integrated $L_\infty$ norm. 
The open source code is provided in a Github repository at \href{https://github.com/elenawang93/Labeled-Merge-Tree-Transform}{github.com/elenawang93/Labeled-Merge-Tree-Transform}. 
The main idea of the procedure is as follows. 

\begin{enumerate}
    \item Compute $\text{Crit}(G) = \text{Crit}(G_1)\cup \text{Crit}(G_2)$. 
          Retain a list of \emph{key angles}, i.e., midpoints between each pair of adjacent critical angles.
    \item Create the labeling scheme of Sec.~\ref{ssec:labelingscheme} for a pair of input graphs $G_1$ and $G_2$ by determining the collection of extremal vertices of each.
    \item Build the merge tree for each key angle and for  each vertex of the tree, associate  the vertex of $G$ which determines its function value.
    \item Use this structure to compute the piecewise-continuous function for every entry in the matrix $\omega \mapsto |M_1^\omega - M_2^\omega|$.
    \item For any $\omega$, find the maximum of these entries and integrate the resulting  $L_\infty$ norm distance output to obtain the LMTT distance.
\end{enumerate}

We highlight the key algorithms in more detail while providing a run-time analysis of these sections with some commentary on decisions and future work. 
To fix notation, we assume we have graphs $G_1 = (V_1, E_1)$ and $G_2 = (V_2, E_2)$, with $n_i = |V_i|$ and $m_i = |E_i|$ for $i = 1,2$, and let $n = n_1 + n_2$ and $m = m_1 + m_2$.  As these are embedded graphs in $\R^2$, recall that Euler's formula implies $m = O(n)$.

Step 1 takes time $O(m)$ since we need only keep a list of normal vectors to all edges. 

We break up the running time computation for Step 2 into two pieces: finding the extremal vertices, and then using these to construct the labeling scheme. 
We use the \texttt{PointPolygonTest} in OpenCV \cite{opencv_library} to find the sets of extremal vertices in each graph. 
Determining if a vertex $v$ is extremal runs in $O(\deg(v))$ where $\deg(v)$ is the degree of the vertex.
Thus checking all vertices for whether they are extremal runs in time $O(\sum_{v \in V_1 \cup V_1}\deg(v)) = O(m)$. 
Let $N$ be the total number of extremal vertices in the two graphs, so that the label set is $[N]$. 
Note that $N \leq n$. 

For the labeling scheme, we take every extremal vertex in $v \in V_1$ and compute its distance to every edge $e \in E_2$. 
If the nearest point is on an edge, we subdivide the edge by adding a vertex with corresponding labels.
If it is a vertex, we add a label to the existing vertex.
We repeat this process for each extremal $v \in V_2$ to find their nearest points on $G_1$.
This requires up to $(n_1(n_2 + m_2)+n_2(n_1+ m_1)) = O(mn)$ distance calculations. 
In total, this means that Step 2 takes time $O(m + mn) = O(mn)$. 
At most, we add $N\leq n$ additional vertices and edges by adding a degree two node that is not already a vertex.

For Step 3, we compute merge trees using the UnionFind data structure to maintain connected components, following the approach in~\cite{Smirnov2020}, which takes $O(m_i\log n_i)$ steps in the worst case on a graph with $n_i$ vertices and $m_i$ edges.
This is done for  $|\text{Crit}(G)| = 2m$ directions, making a running time of $O(m^2\log n)$ for Step 3. 

Processing the merge trees to find the LCA matrix leverages Tarjan's off-line algorithm~\cite{Tarjan1979}, as implemented in NetworkX \cite{NetworkX}, resulting a running time of $O(n^2)$ in each direction.
Again because there are $2m$ directions, this means Step 4 takes time $O(mn^2)$.

Putting this together, we have that Steps 1-4 take time 
$O(m + mn + m^2 \log n + mn^2)$; using Euler's formula, this yields a total running time of $O(n^3)$.
See Fig.~\ref{fig:KDS} for an example of the output after these steps. 
In this figure, at right we have a collection of functions giving the entries for each pair of labels; i.e. for each entry in the difference matrix.
Then the goal for the final step is to determine the maximum value of these functions over $\omega$. 
We next provide two computation methods for Step 5: an exact but slower computation using kinetic data structures and a faster approximation method. 

\begin{figure}
    \centering
    \includegraphics[width = \linewidth]{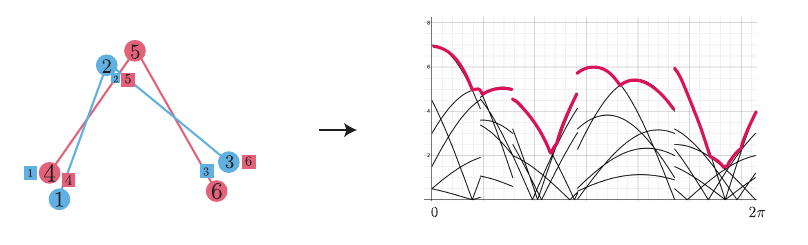}
    \caption{Maintaining maximum value for computing LMTT}
    \label{fig:KDS}
\end{figure}

\subsection{Exact Distance}
To compute the exact distance, we implemented a kinetic data structure (KDS). 
We first introduce the basic framework for a KDS, then describe how our problem fits in the framework. 
For more details on KDS, see \cite{KDSBaschGuibas, KDSGuibas1}.

A KDS maintains a system of objects $v$ that move along a known continuous \emph{flight plan} as a function of time, denoted as $v = f(t)$.
\emph{Certificates} are conditions under which the data structure is accurate, and \emph{events} track when the certificate fails and what updates need to be made correspondingly.
The certificates of the KDS form a proof of correctness of the current configuration function at all times.
Updates are stored in a priority queue, keyed by the event times.
Finally, to advance to time $t$, we process all the updates keyed by times before $t$ and pop them from the queue afterward until $t$ is smaller than the first event time in the queue. 

When computing the maximum entry of the difference matrix $\widetilde M^\omega := |M_1^\omega - M_2^\omega|$, we can construct the problem as the kinetic maximum maintenance problem.
The state-of-the-art algorithm incorporates a heap structure \cite{KDSHeapBasch, KDSHeapTree}.
Each entry of the matrix can be considered as an object moving along a flight plan with a parametrization of the form $y(\omega) = a\sin{\omega} + b\cos{\omega}$, where $a$ and $b$ are constants.
The certificate is a maximum heap that accurately represents the relations between the values of the entries in the matrix.
The certificate fails when two curves cross each other, and the failure times are the angles at which the crossings happen. 
Since all of our objects have the same frequency of $2\pi$, each pair of curves crosses exactly twice in $[0, 2\pi]$. 
Explicitly, for $v_1 = a_1\sin{\omega} + b_1\cos{\omega}$ and $v_2 = a_2\sin{\omega} + b_2\cos{\omega}$, they cross at ${\omega_1 = \arctan\left({\frac{a_1-a_2}{b_1-b_2}}\right)}$ and $\omega_2\equiv\omega_1\pmod{\pi}$. 
The $\omega$ for which any crossing happens are the queue keys in the event queue, while the queue values are the two intersecting objects.
We perform the update by swapping the locations of the two objects in the maximum heap. 
Finally, by maintaining the maximum object, we find the distance function and integrate it to obtain the final distance.
See Figure \ref{fig:KDS} for an example.

This approach is implemented in our code by maintaining a kinetic heap of $n^2$ objects between each pair of $2m$ critical angles.
Consequently, the complexity for this computation is $O(2m\cdot n^4\log n^2)=O(n^5\log n)$ \cite{KDSBaschGuibas}. 
This results in an overall complexity of Steps 1-5 of 
$O(m + mn + m^2 \log n + mn^2 + n^5\log n)$, which simplifies to $O(n^5\log n)$ for the deterministic algorithm.

This can be further optimized to $O(2m\cdot \lambda_2(n^2)\log n^2) = O(n^3\log^2 n)$ by using a kinetic hanger \cite{Agarwal2000}, which introduces randomness to the algorithm. 
Here, $\lambda_2$ is the length bound of a Davenport–Schinzel sequence, where the subscript $2$ indicates the objects are $2$-intersecting curves.  The final expected running time for Steps 1-5 is then $O(n^3\log^2 n)$ using the kinetic hanger data structure.

\subsection{Approximated Distance}
Although the previous section gives a computation for the exact distance, it leads to a practical slowdown. 
While there is code for an exact method in our repository, we use the following sampling method in practice, as it performs significantly better on our data sets.

Once we computed the LCA matrix for each key angle, we check for the maximum entry in the difference matrix, which takes time $O(n^2)$.
We sample to get $\omega$ in $K$ directions, and use the trapezoid rule to integrate. 
While the worst running time in this approach  is  $O(m + mn + m^2 \log n + mn^2 + n^2K) = O(n^3 + n^2K)$, shaving off only a $\log^2 n$ factor in the worst case, it runs significantly faster in practice compared to the implemented kinetic heap. 

We quantify the quality of this approximation by the trapezoidal rule error analysis. 
Given function $F$, sampled at $K$ evenly spaced points over the interval $[A, B]$, the error $E_T$ can be approximated by 
\[
\abs{E_T} \leq \frac{(B-A)^3}{12K^2}\max(F'').
\]
Since the LMTT distance function $F(\omega)$ is piecewise continuous and is composed of functions in the form of $y(\omega)=a \sin\omega + b \cos\omega$, we have $\max(y''(\omega))= \sqrt{a^2+b^2}$. 
Given $a = v_1 - u_1$ and $b = v_2 - u_2$ for any two vertices $v = (v_1, v_2)\in G_1$ and $u = (u_1, u_2)\in G_2$, and let $R$ be the radius of the bounding circle that contains both $G_1$ and $G_2$.
This gives $\sqrt{a^2+b^2}\leq 2R$.
Therefore, we can bound the second derivative  by $\max F''(\omega)\leq 2R$.
Hence, we have a coarse bound given by
\begin{align*}
    \abs{E_T} &\leq \frac{\max(F'')(B-A)^3}{12K^2}
    =\frac{2R(2\pi)^3}{12K^2}
    =\frac{4}{3}\cdot\frac{R\pi^3}{K^2}.
\end{align*}

\section{Applications}
\label{sec:application}

\begin{figure}
\begin{center}
    \includegraphics[width=0.8\textwidth]{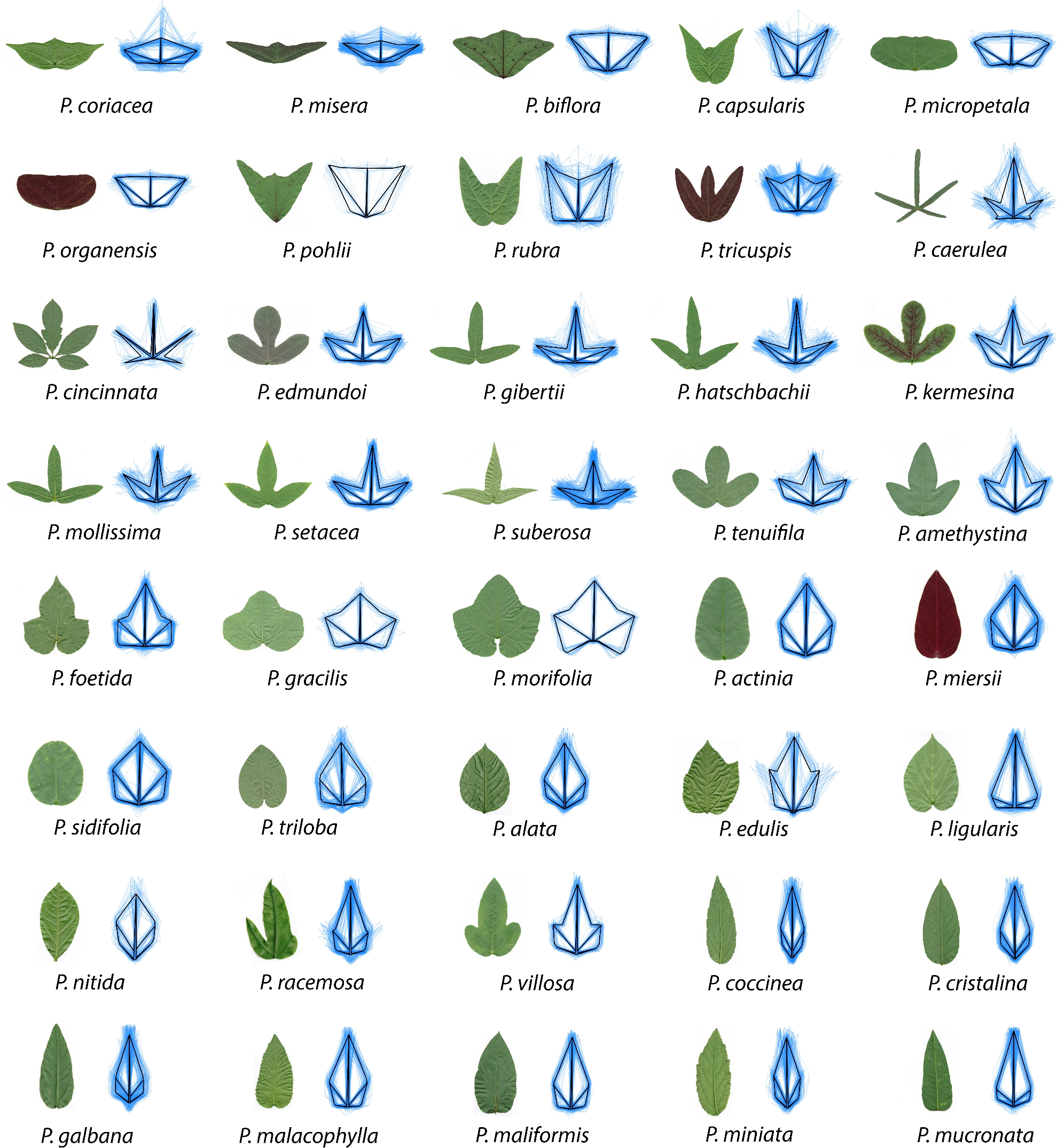}
    \caption{The shapes of \emph{Passiflora} leaves measured using landmarks \cite{Chitwood2017}. \label{Passiflora}}
\end{center}
\end{figure}
\label{sec: Application}
To test the meaning of our distance in practice and its implementation, we use it to visualize the structural differences in two different data sets.

\subsection{Isomorphic graphs: \emph{Passiflora} leaves}
The first data set we examine consists of \emph{Passiflora} leaves \cite{Chitwood2017}.
This dataset contains $3319$ leaves from 226 plants across $40$ species which were classified into seven morphotypes based on traditional morphometric techniques. 
Each entry represents a leaf, with $15$ Procrustes-aligned $(x, y)$-coordinate pairs representing the locations of $15$ landmarks. 
We construct a graph for each leaf by connecting adjacent landmarks, creating an outline that approximates the shape of the leaf. The resulting graphs consist solely of extremal vertices, making this data set an efficient use case of our metric.
See \figurename~\ref{Passiflora} for a comparison of each species' canonical leaf shape with the average graph.
Examples of the seven morphotypes for classification can be seen in the center of Fig.~\ref{fig:MDSPCA}. 

Biologically, \emph{Passiflora} leaves are of great interest to biologists due to their shape diversity while being in the same genus. 
Also of interest is the relationship between leaf shape and maturity. 
As leaves develop along the vine, their shape changes dramatically from juvenile leaves to mature leaves to aged leaves; this phenomenon is known as \emph{heteroblasty}. 
Young leaves from plants of different species often look more similar to each other than to mature leaves from the same plant. 
In order to select leaves that are the most representative of their species, we sample one leaf from each of the $226$ plants that represents the mature shape, with a heteroblasty number of around 0.5.

The results of this analysis are shown in \figurename~\ref{fig:MDSPCA}.
First, on the right of Fig.~\ref{fig:MDSPCA} we have reproduced the analysis of the leaves from \cite{Chitwood2017}. 
Treating each data point as a 30-dimensional vector (which is only possible because of the landmark labels), the authors of \cite{Chitwood2017} used Principal Component Analysis as the dimension reduction algorithm. 
While this plot shows some separation between morphotypes, the overlapping points make the visualization difficult to interpret. 

We then compute pairwise LMTT distances between all the leaves (note that this does not utilize the landmark labels in any way) and construct a $226\times 226$ distance matrix.
In order to  to visualize the relationships given by this distance, we use Multi-Dimensional Scaling (MDS) \cite{Cox2008} which is shown in the left of Fig.~\ref{fig:MDSPCA}. 

Comparing the LMTT to the PCA, we see a similar global structure but with better separation amongst points.
However, we note in particular that the LMTT comparison is done with no memory of the landmark structure, meaning we have similar results with less input information. 
We note that most of the overlapping in the MDS plot occurs amongst the leaves of morphotypes E, F, and G, which have very similar outlines as seen in the middle examples of Fig.~\ref{fig:MDSPCA}. 
Amongst leaves with sufficiently different outlines, such as type A, B, C, and D, we see the distinction reflected by the embedding.
This leads us to believe that the LMTT is a good tool for use in distinguishing shapes, so long as they are sufficiently distinct.

\begin{figure}\centering
  \includegraphics[width=\textwidth]{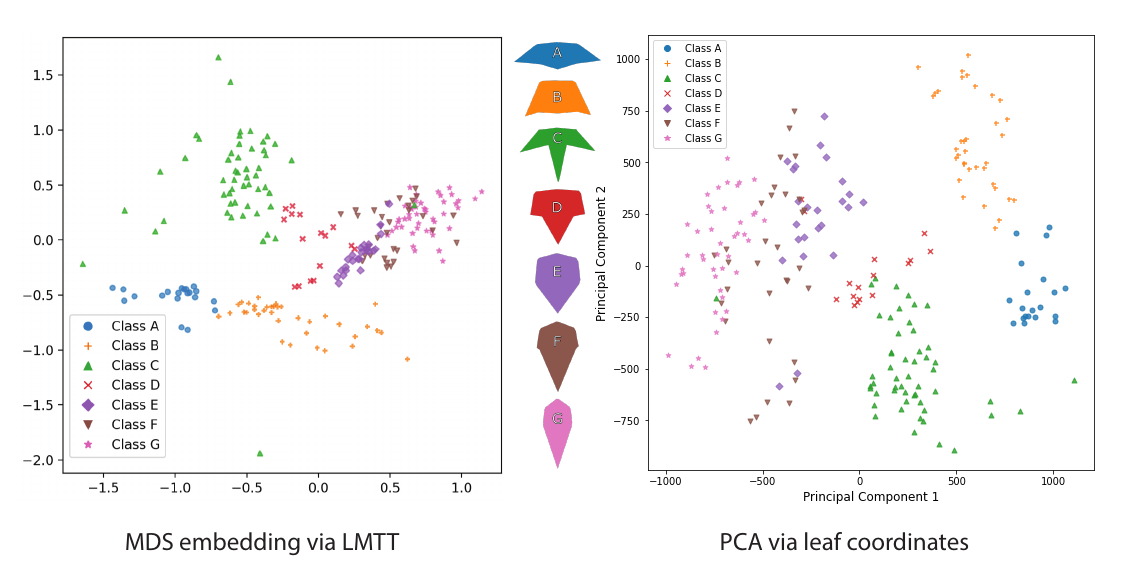}
  \caption{MDS and PCA plots with points colored by morphotype. 
  \label{fig:MDSPCA}}
\end{figure}

\subsection{Non-isomorphic graphs: letters}
The second data set we examine is a set of different letters from the IAM graph database \cite{Riesen2008}.
However, there are example letter graphs in this data set which are quite distorted; therefore, we picked the data points with only one connected component. 
Further, the dataset consists of some unrecognizable data points since they are generated from hand-written letters.
Therefore, we created a template graph for each letter to remove these outliers from the data before running the experiment.
We removed data points if they did not have the same number of vertices and the same edge list as their letter's template. 
The data set we end up with has 450 data points: 15 different letters, and 30 graphs per letter.
Each letter graph has 2 to 7 vertices and 1 to 6 edges.
We follow the same process as in the leaf shapes by computing all pairwise LMTT distances and constructing a $450\times 450$ distance matrix.

The relationship between the letters is then visualized via MDS; see Fig \ref{fig:letter150}.
We observe clear separation of clusters amongst the different letters. 
The nearby clusters also follow what we expect, given the choice of merge tree as our topological signature. 
For example, the letters ``H" and ``X" can generate similar merge trees in certain directions, and the MDS plot shows those clusters close together.
The graphs in this data set differ from the leaf example in that they have different numbers of vertices and edges; hence, they are non-isomorphic.
This shows the utility and meaningfulness of our distance on two datasets that have different isomorphism properties.

\begin{figure}
    \centering
    \begin{minipage}{0.43\textwidth}
        \includegraphics[width = \textwidth]{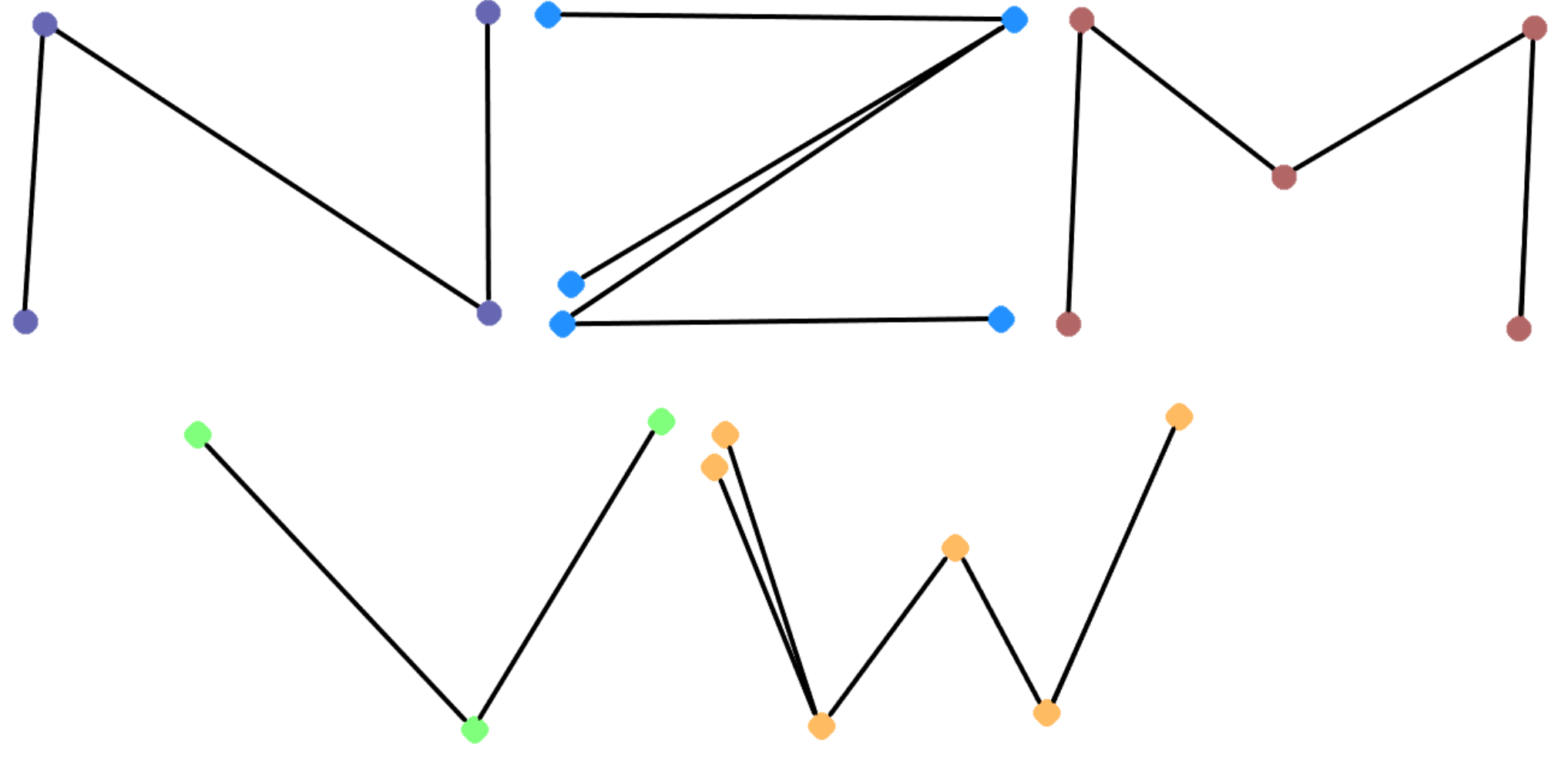}
    \end{minipage}
    \begin{minipage}{0.53\textwidth}
        \includegraphics[width = \textwidth]{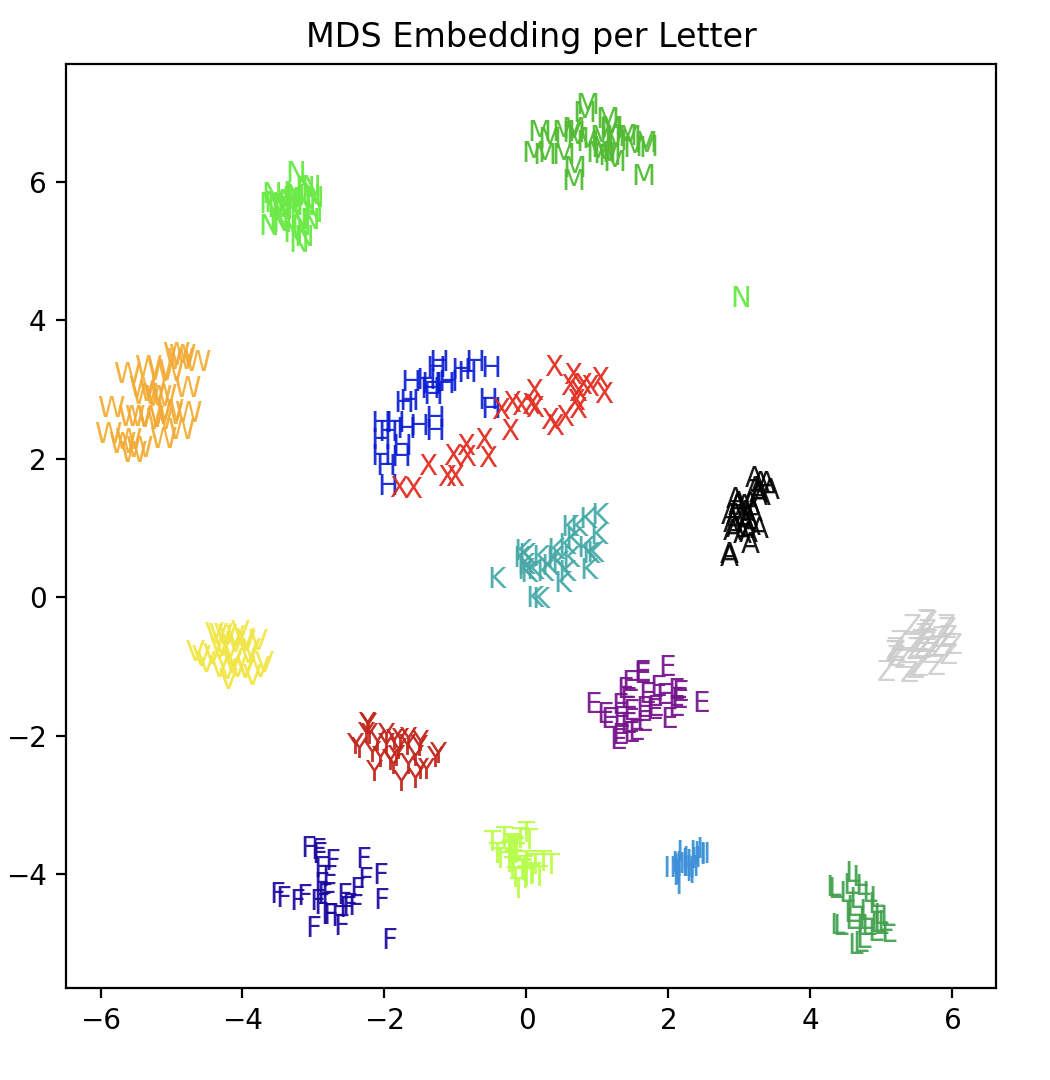}
    \end{minipage}
    \caption{At left, an example of letter graphs from the IAM data set. At right, the MDS plot labeled by letter computed using the LMTT distance.}
    \label{fig:letter150}
\end{figure}

\section{Conclusion and Discussion}
In this paper, we have given a new distance between embedded graphs by constructing a labeling procedure and computing distance based on the labeled merge tree interleaving distance. 
We proved theoretical properties of the labeled merge tree transform distance and used them for optimization in its implementation. 
We showed its utility by using it to discriminate among leaves in the \emph{Passiflora} data set and between letter graphs in the IAM data set.

Our method comes with its own strengths and limitations. 
As a strength, it is easy to compute compared to many embedded graph distances; see \cite{Buchin2023} for further details on other options. 
In the future, we aim to optimize the algorithm and implementation to allow for applications to larger data sets. 
However, while it summarizes the outer shape of the data and differentiates pairs of data points that are similar, we note that it cannot distinguish the internal structures of data. 
For example, the distance between a hexagon and hexagon with a central vertex as in Fig.\ref{fig:InternalVertex} is 0 since the merge trees are the same for the two graphs in any direction.
\begin{figure}
    \centering
    \includegraphics[width=0.7\linewidth]{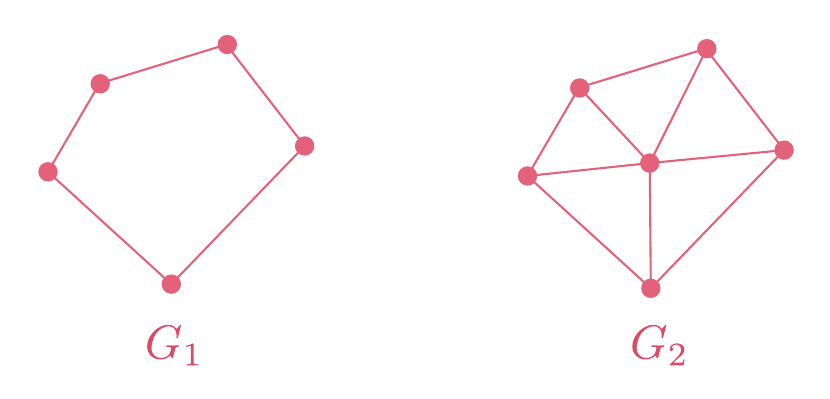}
    \caption{An example of graphs with LMTT distance 0.}
    \label{fig:InternalVertex}
\end{figure}
Another limitation of the method is that it requires the data to be pre-aligned, a difficult problem in its own right.
Taken in another light, however, the distance is based on the embedding information so we would expect that alignment of the data is taken into account when computing this distance. 

We also note that there is an instability inherent to the labelling procedure. 
See the example of Fig.~\ref{fig:LabelSwitch} where a slight edit to the location of the vertex $w_5$ from $G_2$ to $G_2'$ results in the label jumping from one edge of $G_1$ to the other. 
This occurs because the vertex is moved past the bisector of the angle at $v_1$.
For this reason, it would be interesting to see if there are other, more stable methods for labeling which can be used in the procedure while still resulting in similar distance properties.
While this example shows that stability for the entire collection of embedded graphs is unlikely, it would be interesting to see if this distance is stable for restricted sets. 
These restricted settings may also allow for stronger properties of the distance such as separatbiliy and the triangle inequality. 

\begin{figure}
    \centering
    \includegraphics[width = 0.75\linewidth]{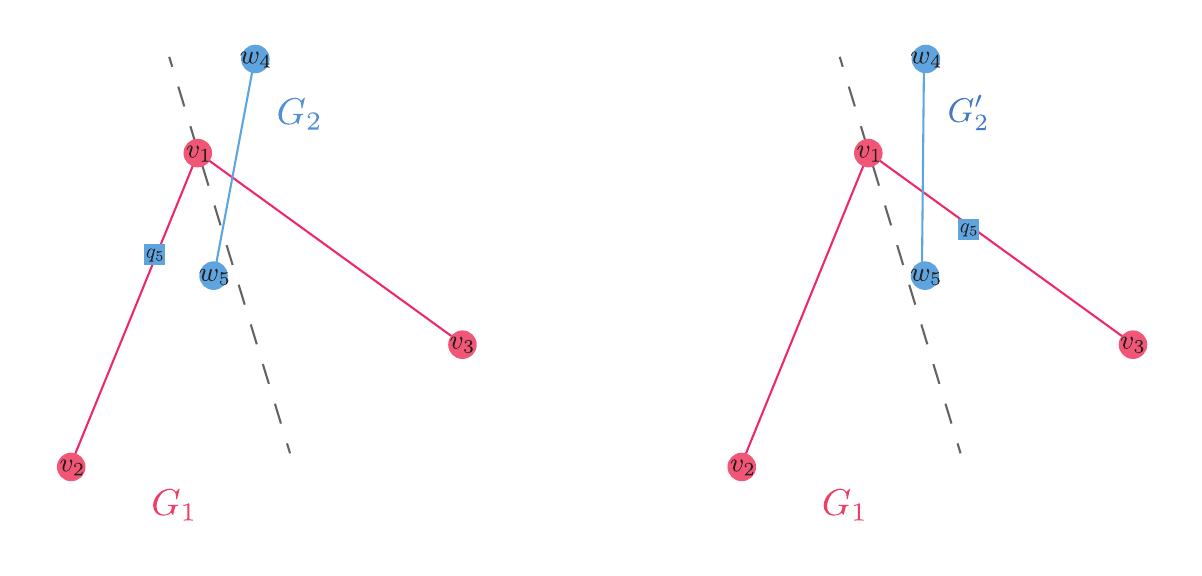}
    \caption{An example of a small change in a graph resulting in a large change in labeling locations.}
    \label{fig:LabelSwitch}
\end{figure}

Finally, in theory this definition should be applicable to graphs and/or simplicial complexes embedded in higher dimensions $\R^d$. 
The labeled merge tree transform would result in collections of matrices parameterized by $\S^{d-1}$ for higher dimensional spheres, so further care would need to be taken to understand how to track the maximum value of the distance for the integral definition.

\section*{Acknowledgements}
XW would like to thank Vin de Silva for helpful discussions.
This work was funded in part by the National Science Foundation through grants 
CCF-2106672,
CCF-1907612,
CCF-1907591,
CCF-2106578,
and CCF-2142713.

\bibliographystyle{plainurl}
\bibliography{SoCGYRF}

\end{document}